%% file: main.tex
\documentclass[conference]{IEEEtran}

\usepackage{cite}

\usepackage{amsmath,amssymb,amsfonts}
\usepackage{mathtools}
\usepackage{amsthm}
\usepackage{algorithmic}
\usepackage{graphicx}
\usepackage{textcomp}
\usepackage{xcolor}
\usepackage{booktabs}
\usepackage[hidelinks]{hyperref}
\usepackage{url}
\usepackage{subcaption}
\usepackage{cleveref}
\usepackage{tikz}
\usepackage{verbatim}
\usetikzlibrary{arrows.meta}

\graphicspath{{figs/}}

\newtheorem{definition}{Definition}

\newtheorem{theorem}{Theorem}
\newtheorem{proposition}{Proposition}

\newtheorem{remark}{Remark}
\newtheorem{property}{Property}
\newtheorem{example}{Example}

\crefname{definition}{Definition}{Definitions}
\Crefname{definition}{Definition}{Definitions}
\crefname{assumption}{Assumption}{Assumptions}
\Crefname{assumption}{Assumption}{Assumptions}
\crefname{lemma}{Lemma}{Lemmas}
\Crefname{lemma}{Lemma}{Lemmas}
\crefname{theorem}{Theorem}{Theorems}
\Crefname{theorem}{Theorem}{Theorems}
\crefname{proposition}{Proposition}{Propositions}
\Crefname{proposition}{Proposition}{Propositions}
\crefname{corollary}{Corollary}{Corollaries}
\Crefname{corollary}{Corollary}{Corollaries}
\crefname{remark}{Remark}{Remarks}
\Crefname{remark}{Remark}{Remarks}
\crefname{property}{Property}{Properties}
\Crefname{property}{Property}{Properties}
\crefname{appendix}{Appendix}{Appendices}
\Crefname{appendix}{Appendix}{Appendices}

\makeatletter
\renewenvironment{proof}[1][\proofname]{\par
  \pushQED{\qed}%
  \normalfont \topsep6\p@\@plus6\p@\relax
  \trivlist\item[]%
  {\itshape #1\@addpunct{.}}\enspace\ignorespaces}
  {\popQED\endtrivlist\@endpefalse}
\makeatother

\begin{document}

\title{Compatibility and future-proof testing for multi-agent systems\\
}

\author{
\IEEEauthorblockN{Yueshan Li, Lorenzo Fiaschi, Andrea Censi, and Emilio Frazzoli}
\IEEEauthorblockA{\textit{Institute for Dynamic Systems and Control} \\
\textit{ETH Zurich}\\
Zurich, Switzerland}
}

\maketitle
\begin{abstract}
Independently developed autonomous agents, such as autonomous vehicles, must guarantee joint compatibility, such as safety and task-completion requirements, when they operate in a shared environment.
Existing methods require a centrally designed specification, or interaction mechanism before any agent is deployed.
However, such centralized design is unrealistic when agent designs are confidential to their developers, and it is not flexible enough to accommodate future agent designs.
We study a bi-level structure in which a regulator coordinates developers by publishing seed agents,
and each developer independently tests its agent to ensure that it is compatible with every seed-compatible agent, a property we call future-proofness.
Under this structure, joint compatibility is guaranteed for any agents that are designed and tested independently, without any coordination with each other.
In this work, we first formulate compatibility for general multi-agent systems
and show why it is not trivially implied by seed-compatibility.
We then characterize the future-proof agents as the closure of the seed and show that this closure is monotone in the seed.
Finally, we study a class of executable tests that involve interaction with the seed alone, and we ask to what extent such tests can distinguish the future-proof agents.
\end{abstract}

\begin{IEEEkeywords}
autonomous agents, multi-agent systems, certification, safety
\end{IEEEkeywords}

\input{sections/introduction}
\input{sections/related-work}
\input{sections/framework}

\input{sections/case-study}
\input{sections/conclusion}
\input{sections/appendix}

\bibliographystyle{IEEEtran}
\bibliography{references}

\end{document}

%% file: sections/introduction.tex
\section{Introduction}
\label{sec:introduction}
\begin{figure}[!t]
    \centering
    \includegraphics[width=\columnwidth]{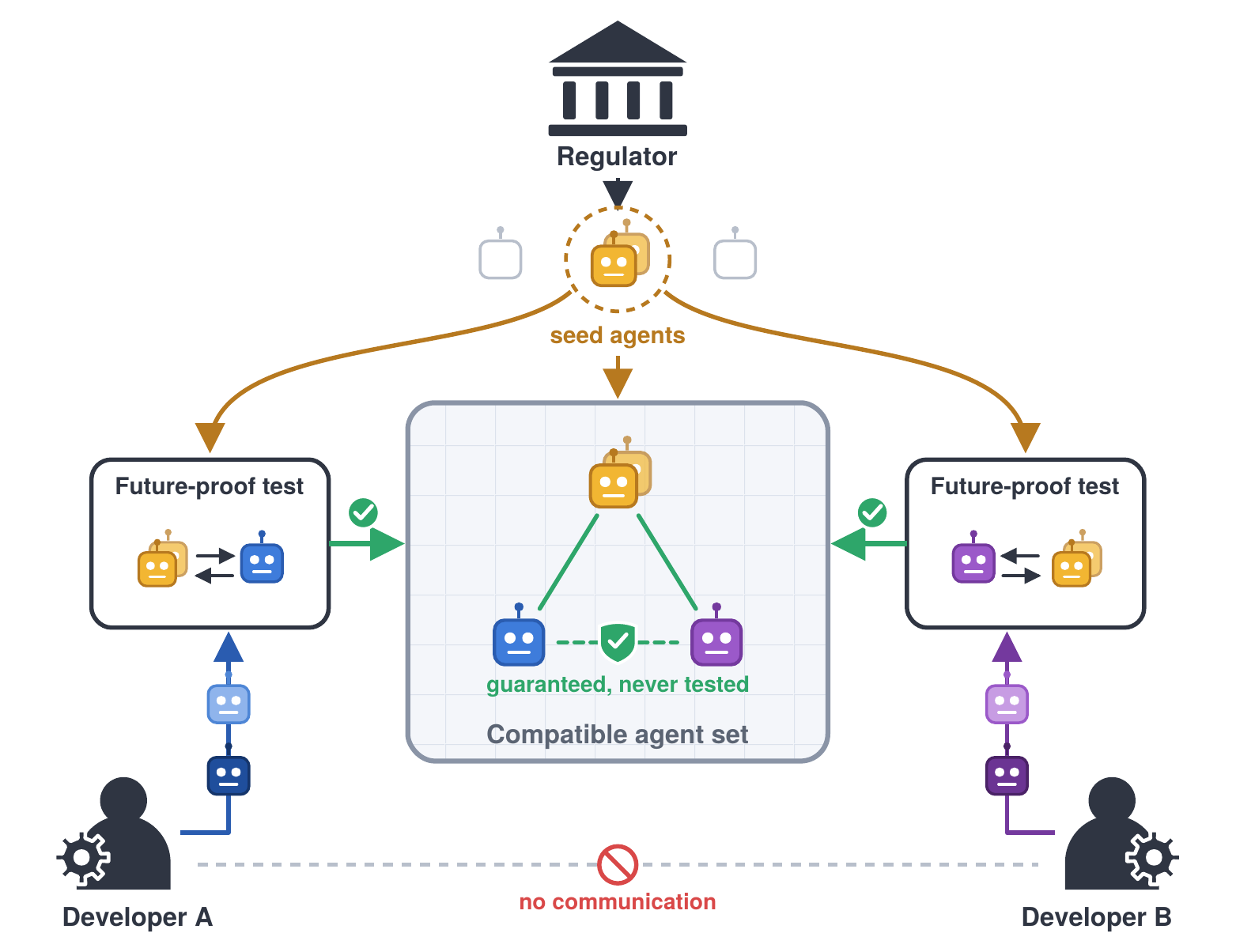}
    \caption{Scheme of future-proof testing. A regulator publishes the seed agents (top). 
    Each developer tests its own agents against the seed alone (sides), with no communication between developers. 
    Any agent that passes joins the compatible agent set (centre). It is compatible with every passing agents, although they are never tested against each other.}
    \label{fig:overview}
\end{figure}

Autonomous agents increasingly take part in daily life and share their environment with human beings and with other autonomous agents.
These agents serve different purposes, come from different vendors, and are designed independently. Guaranteeing compatibility among them is therefore a challenge.
For instance, as many as 26 companies have obtained permits to test automated vehicles on public roads in California~\cite{cadmv2026permits},
and drone operators fly through overlapping airspace~\cite{astm2021f3548}.
Beyond robotics, large-language-model (LLM) agents are expected to cooperate with human beings and with each other~\cite{hammond2025multiagent}.
In all these cases, a population of autonomous agents is deployed by different developers without direct coordination.
Nevertheless, their joint behavior must remain acceptable, so that safety and efficient task completion are guaranteed.
In short, they must be compatible.

Compatibility is typically certified through centralized design.
For instance, in aviation a conflict-resolution protocol is imposed on every participant~\cite{faa2011tcas},
and safety is guaranteed as long as all participants follow the protocol.
However, centralized design is unrealistic when agent designs are confidential to their developers
and regulators have no right to pre-approve them~\cite{nhtsa2020selfcert}.
Moreover, agents can be updated over the air after deployment, and a single update can change two million vehicles overnight~\cite{nhtsa2023recall838}.
As a result, new testing concepts are needed that rely neither on the implementation details of the agents nor on coordination between developers.

We propose \emph{future-proof testing}, illustrated in~\Cref{fig:overview}.
A regulator publishes a set of agents that serves as the common reference for all developers, which we call the \emph{seed}.
The seed is the set of agents that the regulator want to be compatible with all the future ones. For instance, this set can contain the agents already deployed, or represent the conventions that the regulator wants to see in their system.
Each developer then tests its agent independently against the seed, without coordination with other developers.
Once an agent passes the test, it is compatible with the seed and with every other passing agent, even though the passing agents are never tested against each other.

Seed-based coordination of this kind is already practiced.
For example, UN Regulation 157 scripts a family of cut-in manoeuvres that an automated lane-keeping system must survive~\cite{unece2021r157}.
However, designing the passing criterion of such a test is non-trivial.
\Cref{sec:motivating} gives a simple counterexample in which passing against the seed does not imply that the certified agents are safe together.

This paper establishes the theoretical foundations of future-proof testing with safety guarantees.
We propose a general formulation of multi-agent systems and acceptance criteria.
On this basis, we define a compatible set as a set of agent designs of which any combination can be deployed without violating the acceptance criteria.
Within this framework, we first show, both by example and in theory, that seed compatibility does not in general imply mutual compatibility.
We then characterize the \emph{future-proof} agents, which are the target of future-proof testing and provide the desired safety guarantees.
Surprisingly, the set of future-proof agents grows with the seed, so a more demanding reference admits more future-proof agents.
Finally, we ask what a developer can establish from testing against the seed alone.
We answer with a lower bound on the testing loss of seed-based testing, that is, the future-proof agents that no such test can identify.


This paper makes three contributions.
\begin{enumerate}
    \item We provide theoretical foundations for seed-based coordination and explain why testing against the seed alone does not naturally lead to compatibility among agents.
    \item We characterize the target of correct seed-based testing, namely the future-proof agents, and bound the testing loss of seed-based testing from below.
    \item We showcase the theory in two case studies: a chain of braking vehicles and a decentralized path finding in a grid-world warehouse.
\end{enumerate}

%% file: sections/related-work.tex
\section{Related Work}

In this section, we review representative works that target at the joint safety of multi-agent systems, which fall mainly into three families:
a posteriori verification, safe-by-construction synthesis,
and individual conformity assessment against a common safety standard.

\textbf{Formal verification of multi-agent systems.}
Formal verification uses mathematical methods to prove that a system satisfies a given specification that encodes safety conditions.
One widely used approach is reachability analysis, which computes the reachable sets of a system from a target set and checks whether they ever violate the specification.
This problem can be formulated as a differential game, whose solution is characterized by the Hamilton--Jacobi (HJ) equations~\cite{mitchell2005hj}.
The formulation is theoretically sound, but solving the HJ equations is computationally demanding and donot scale well with the number of agents~\cite{bansal2017hj}.
Furthermore, since the specifications of multi-agent systems inherently depend on the joint state of multiple agents,
the decomposition techniques~\cite{chen2018decomposition} that apply to many single-agent systems do not extend to general multi-agent systems.
The method has therefore been demonstrated only on small-scale multi-agent systems~\cite{10160600} or
on large-scale systems with special structure, such as sequential trajectory planning~\cite{bansal2021provably}.

To overcome the computational complexity of HJ reachability,
set propagation~\cite{althoff2021set} approximates the reachable sets by propagating the initial set through the system dynamics.
With specific set representations and system dynamics, set propagation can run in real time
and has been applied to the online verification of autonomous vehicles against other road users~\cite{althoff2014online}.
The method is limited to short-horizon verification, however, since the sets are only over-approximations of the actual reachable states and their error compounds over time.

Beyond computational complexity, formal verification is a posteriori:
it requires knowledge of the mathematical model of each agent, as well as of how the agents interact.
Such methods are not applicable to our setting, where agents must be tested independently and must guarantee safety alongside agents that are unknown at testing time.

\textbf{Synthesis of safe multi-agent systems.}
Synthesis methods provide a
multi-agent system with safe-by-construction guarantees.
One synthesis approach is the control barrier function (CBF) method~\cite{wieland2007cbf}, which encodes
a control-invariant set and overwrites any action that would drive the system out of it.
For control-affine systems, the CBF can be cast as a quadratic program and solved efficient for
both single-agent systems~\cite{ames2014acc,ames2017cbf} and centralized multi-agent
systems~\cite{wang2017safetybarrier,pickem2017robotarium}.
Decentralizing the CBF-based controllers, however, is nontrivial.
It typically requires hand-crafted responsibility splitting rules to decompose the coupled safety constraints into local ones~\cite{wang2017safetybarrier,lindemann2019decentralized,zhang2025deadlock},
and such rules face an increasingly severe trade-off between conservatism and infeasibility as the number of agents grows.

Instead of committing every agent to a prescribed control law, an alternative is to
coordinate their behaviors through behavioral contracts.
Originating in assume-guarantee reasoning for concurrent
software~\cite{jones1983rely,henzinger1998assume}, a behavioral contract specifies how an
agent should behave given certain behaviors of the others~\cite{saoud2021continuous,sharf2021contracts,dewes2025contract}.
The constracts are carefully designed such that once satisfied by all agents, the joint safety of the system is guaranteed.
For instance, Responsibility-Sensitive Safety ~\cite{shalevshwartz2017rss} prescribes the minimum distances a vehicle must maintain
from other road users under worst-case assumptions on their kinematics.
In~\cite{cai2023rules}, each vehicle is required to follow a hierarchy of traffic rules under the assumption that all others are homogeneous and follow the same rules.
In~\cite{liu2025stl}, the contract confines each agent to a predefined envelope specified in signal temporal logic.

Although behavioral contracts scale better than CBFs and admit a wider range of compliant agent designs,
they share the drawback of pre-committing to a specific conflict-resolution mechanism.
As systems grow more complex and agents are expected to exhibit richer behaviors for more demanding tasks, designing contracts that jointly guarantee safety becomes increasingly difficult.
Moreover, how the choice of contract affects the permissiveness of the resulting multi-agent system is not well understood.
Furthermore, synthesis must take place before the agents are deployed, which makes these approaches less applicable to agents already in the field, such as existing human drivers on roads.

\textbf{Standards for joint safety.}
A technical standard specifies requirements that products are expected to satisfy before entering the market. 
While most requirements and the associated conformity assessments concern the performance of a product in isolation,
standards in several fields have recently begun to address the joint behaviors of products that interact with each other.

For example, the interconnection standard for distributed energy resources (DERs)~\cite{ieee1547-2018} specifies
the interface behavior each DER must exhibit, such as voltage regulation and anti-islanding,
with the aim of keeping the grid stable and safe as conformant DERs are interconnected.
The conformity assessment of each DER is carried out independently against a common, standardized set of grid conditions~\cite{ieee1547.1-2020}.
Although effective in practice, the prescribed interface behaviors and test conditions are derived from power-system analysis and established operating limits,
and come with no formal guarantee that the resulting interconnected system is safe.

IEEE Std 2846-2022~\cite{ieee2846-2022}, which concerns assumptions in safety-related models for automated driving systems,
represents a step toward incorporating interaction into autonomous vehicle standards.
It specifies a minimum set of reasonably foreseeable assumptions on the kinematic behavior of other road users that the safety-related model of the ego vehicle must account for.
However, the standard prescribes neither a conformity assessment procedure nor the behavior expected of the ego vehicle,
and it provides no guarantee that conformant vehicles remain safe when interacting with one another.

These conformity assessments resemble our notion of future-proof testing in form:
each agent is assessed independently against a common set of reference setups,
and the requirements are imposed on agent behaviors rather than its implementation.
However, there is still no theoretical foundation for choosing the reference setups and for deciding which agents to accept so that joint safety is guaranteed.
Our theory fills this gap by formally characterizing the agents that the tests should admit in order to guarantee joint safety,
and by showing how the choice of reference affects both the set of agents that can in be admitted and the difficulty of identifying them.

\color{black}

%% file: sections/framework.tex
\section{Future-proof compatibility testing}
\color{black}
\label{sec:compatibility}
In this section, we propose a general framework for future-proof testing.
We first show that safety with the seed is not sufficient for certification, because agents that are each safe with the seed are not necessarily safe together.
We then characterize exactly which agents are future-proof, and show that enlarging the seed makes more agents future-proof, although it constrains developers more.
Finally, we bound the testing loss of seed-based testing, that is, the future-proof agents whose future-proofness a developer cannot establish from the seed alone.

We first illustrate these findings on a simple example of points in the plane, and then develop the theory.

\subsection{Introductory example}
\label{sec:motivating}
\color{black}
Consider agents that are points in the plane.
A set of points is safe if any two of them are at most $\delta$ apart.
The regulator publishes a set of seed points.
Each developer then places a point knowing only the seed.
Nevertheless, all points, including the seed, must be safe together.

\textbf{Safe with the seed, yet unsafe together.}
Let the seed be a single point $p_0$.
The points that are safe with $p_0$ form a disk of radius $\delta$ centered at $p_0$, shown as the gray region in~\Cref{fig:motiv}a.
However, two such points may lie on opposite sides of $p_0$, up to $2\delta$ apart, and hence be unsafe with each other.
Even in this simple, abstract setting, safety with the seed therefore does not imply joint safety, and a stricter certification criterion is needed.
\Cref{sec:compat-formulation} identifies the root cause of this phenomenon for general multi-agent systems.

\textbf{More seed points, more future-proof agents.}
An agent is future-proof if it is safe together with the seed and with any group of agents that is itself safe with the seed.
For the single seed $p_0$, only $p_0$ itself is future-proof.
Indeed, for any other point $x$, there is a point at distance $\delta$ from $p_0$ on the side opposite to $x$, which is safe with the seed yet unsafe with $x$.
Adding a second point to the seed, as shown in~\Cref{fig:motiv}b, shrinks the set of points that are safe with the seed.
However, more points become future-proof, because each of them now needs to be safe with fewer seed-safe points.
As more points are added to the seed, the set of seed-safe points shrinks while the set of future-proof points grows, until the two eventually coincide.
\Cref{fig:motiv}c shows this extreme case, in which every point that is safe with the seed is future-proof.
This is the best case for a given seed, since no point that is safe with the seed is rejected.
\Cref{sec:closure} formally characterizes the future-proof agents and shows that the same trend holds for general multi-agent systems.

\textbf{Checking against the seed alone.}
In this example, whether a point is future-proof can be decided from its distances to the seed points alone,
without knowing how many other points will be placed or where.
With the single seed $p_0$, a point is future-proof if and only if it coincides with $p_0$.
With two seed points, membership in the future-proof region (purple in \Cref{fig:motiv}b) is likewise determined by the distances to $p_0$ and $p_1$.
In the extreme case, a point is future-proof if and only if it is within $\delta$ of every seed point.
In general, which future-proof agents can be distinguished from non-future-proof ones is a property of the seed.
\Cref{sec:sat} characterizes this property as the separable set of the seed, which lower-bounds the testing loss of seed-based testing.
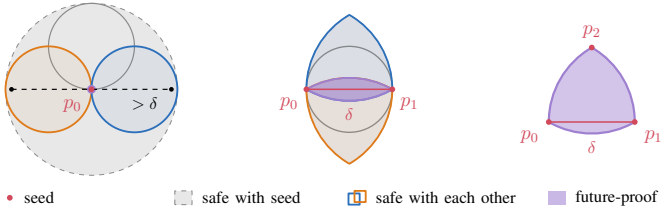
\begin{figure}[t]
  \centering
  \resizebox{\columnwidth}{!}{\input{figs/closure_growth.tex}}
  \caption{Motivating example.
  \textbf{Left.} Two points that are each safe with the seed can be unsafe with each other, and only $p_0$ itself is future-proof.
  \textbf{Middle.} A second seed point shrinks the seed-safe region but enlarges the future-proof region.
  \textbf{Right.} In the extreme case, the seed-safe region coincides with the future-proof region.}
  \label{fig:motiv}
\end{figure}

\color{black}
\subsection{Compatibility Formulation}
\label{sec:compat-formulation}
In this section we first introduce the general formulation of the multi-agent system and the associated notions of compatibility.
We then consider the intended end product of our setup, namely which agents should ultimately be admitted, given a seed, so that joint compatibility is guaranteed.
In particular, we show that checking seed-compatibility alone is not sufficient to guarantee joint compatibility.

An arena collects the four components needed to describe a shared environment of interacting agents, namely
the set of all agents that may enter the system;
the set of all outcomes that can be observed when agents operate together;
an outcome generator specifying which outcomes are observed for a given group of agents;
and the subset of outcomes that are deemed unacceptable.
In a robotics context, for instance, a typical observable outcome is the joint trajectory of the agents, and a natural unacceptable outcome is a collision between them.
The outcome generator can then be instantiated as a simulator that rolls out a given group of agents over a collection of configurations.

\color{black}
We model the input of the outcome generator as an arrangement of agents, which captures the general case in which the environment offers a set of tasks and each agent can be assigned to any one of them.
\begin{definition}[Arrangement]
    \label{def:arr}
    Let $F$ be a set. An \emph{arrangement} of elements from $F$ is a finite sequence $(a_1, \dots, a_n)$ with $n \in \mathbb{N}$ and $a_i \in F$ for all $i \in \{1, \dots, n\}$.
    
    We denote the space of all such sequences from $F$ by $\mathrm{Arr}(F)$, i.e.
    \begin{equation*}
        \mathrm{Arr}(F) \coloneqq \bigcup_{n \in \mathbb{N}} F^n,
    \end{equation*}
    where $F^n$ denotes the $n$-fold Cartesian product of $F$.
\end{definition}
\color{black}
\begin{definition}[Arena]
An \emph{arena} $\mathcal{A} = (\Pi, \mathcal{O}, \Phi, \mathcal{U})$ consists of the following four components.
\begin{itemize}
    \item An agent space $\Pi$.
    \item An observable outcome space $\mathcal{O}$.
    \item {\color{black}An outcome generator $\Phi \colon \mathrm{Arr}(\Pi) \to P(\mathcal{O})$}.
    \item A set of unacceptable outcomes $\mathcal{U} \subseteq \mathcal{O}$.
\end{itemize}
\end{definition}
\begin{example}[Point arena]
\label{ex:point-arena}
The motivating example is captured by the arena $\mathcal{A} = (\Pi, \mathcal{O}, \Phi, \mathcal{U})$ with
\begin{itemize}
    \item $\Pi \coloneqq \mathbb{R}^2$, i.e.\ each agent is a point in the plane.
    \item $\mathcal{O} \coloneqq \mathbb{R}_{\ge 0}$, the possible distances between agents.
    \item {\color{black}$\Phi\big((a_1, \dots, a_n)\big) \coloneqq \big\{ \max_{i,j} \|a_i - a_j\| \big\}$, which returns the largest pairwise distance within an arrangement.}
    \item $\mathcal{U} \coloneqq (\delta, \infty)$, i.e.\ distances exceeding the threshold $\delta$.
\end{itemize}
\end{example}
\begin{remark}
    The arena definition generalizes to non-deterministic systems if the outcome generator
    outputs all outcomes under bounded disturbances, or all outcomes above a certain likelihood under stochasticity.
\end{remark}
\color{black}
A group of agents is said to be compatible only when all arrangements of its members are acceptable.
\begin{definition}[Compatible set]
\label{def:compatible}
A set of agents $F \subseteq \Pi$ is \emph{compatible} in arena $\mathcal{A}$ iff $\forall a \in \mathrm{Arr}(F)$, $\Phi(a) \cap \mathcal{U} = \emptyset$.

We denote the collection of all compatible sets in arena $\mathcal{A}$ as $\mathcal{C}_{\mathcal{A}}$.
\end{definition}
By~\Cref{def:compatible}, if a set of agents is compatible, then any of its subsets is also compatible. This property is stated below and proven in~\Cref{app:proofs}.
\begin{property}[Compatibility is closed under subsets]
\label{prop:subset-closed}
If $G \in \mathcal{C}_\mathcal{A}$ and $F \subseteq G$, then $F \in \mathcal{C}_\mathcal{A}$.
\end{property}


The seed published by the regulator can be any set of agents, as long as it is itself compatible.
\begin{definition}[Seed]
\label{def:seed}
A seed is a compatible set $F_0 \in \mathcal{C}_\mathcal{A}$.
\end{definition}

Given a seed, the ultimate goal of testing is to admit a set of agents that form a compatible set together with the seed.
We establish the basic properties of such sets and leave the proofs to~\Cref{app:proofs}.
\begin{definition}[Compatible sets from a seed]
\label{def:C-F0}
Given a seed $F_0$, we write $\mathcal{C}_\mathcal{A}(F_0)\coloneqq \{ F \in \mathcal{C}_\mathcal{A} \mid F_0 \subseteq F \} \subseteq \mathcal{C}_\mathcal{A}$
for the collection of compatible sets containing the seed.
\end{definition}

\begin{definition}[Maximal compatible set from a seed]
\label{def:maximal-pop}
Given a seed $F_0$,
a compatible set $F \in \mathcal{C}_\mathcal{A}(F_0)$ is \emph{maximal} 
iff $\forall \pi \in \Pi\setminus F$, $ F \cup\{\pi\} \notin \mathcal{C}_\mathcal{A}(F_0)$.
\end{definition}
\begin{proposition}[Existence of a maximal compatible set from a seed]
\label{prop:existence}
    Given a seed $F_0$, a maximal compatible set from $F_0$ always exists.
\end{proposition}
\begin{example}
\label{ex:compat}
In the point arena, any set with diameter at most $\delta$ is compatible, and the seed can be any such set.
The maximal compatible sets from a seed are the bodies of constant width exactly $\delta$ that contain the seed, 
which are in general not unique~\cite{martini2019bodies}.
\end{example}

A straightforward approach to admitting agents is to check whether each of them is compatible with the seed.
In the following we show that this check is insufficient to guarantee joint compatibility.
\begin{definition}[Seed-compatible agents]
\label{def:seed-compatible}
Given a seed $F_0$, the set of \emph{seed-compatible agents} is
\begin{equation*}
    S(F_0) \coloneqq \{ \pi \in \Pi \mid F_0 \cup \{\pi\} \in \mathcal{C}_\mathcal{A} \}.
\end{equation*}
\end{definition}

The following theorem reveals the connection between the seed-compatible agents and the maximal compatible sets from the seed.
\begin{theorem}[Seed-compatible agents form the union of maximal compatible sets]
\label{thm:union}
Given a seed $F_0$, 
\begin{equation*}
    S(F_0) = \bigcup_{F \in \mathcal{M}(F_0)} F,
\end{equation*}
where $\mathcal{M}(F_0)$ denotes the collection of all maximal compatible sets from $F_0$.

Moreover, $S(F_0)\in\mathcal{C}_{\mathcal{A}}(F_0)$ iff the maximal compatible set from $F_0$ is unique.
\end{theorem}
\begin{proof}
    We first show that every $F \in \mathcal{C}_\mathcal{A}(F_0)$ is contained in some $F^\star \in \mathcal{M}(F_0)$.
    Since $F$ is compatible, applying \Cref{prop:existence} to the seed $F$ gives a maximal compatible set $F^\star$ from $F$,
    which is also maximal from $F_0$ because $F_0 \subseteq F \subseteq F^\star$ and every superset of $F^\star$ contains $F$.

    For the equality, let $\pi \in S(F_0)$.
    Then $F_0 \cup \{\pi\} \in \mathcal{C}_\mathcal{A}(F_0)$ is contained in some $F^\star \in \mathcal{M}(F_0)$, so $\pi$ lies in the union.
    Conversely, if $\pi \in F$ for some $F \in \mathcal{M}(F_0)$, then $F_0 \cup \{\pi\} \subseteq F$ is compatible by \Cref{prop:subset-closed}, so $\pi \in S(F_0)$.

    For the uniqueness condition, if $\mathcal{M}(F_0) = \{F^\star\}$, then the equality gives $S(F_0) = F^\star \in \mathcal{C}_\mathcal{A}(F_0)$.
    Conversely, let $S(F_0) \in \mathcal{C}_\mathcal{A}(F_0)$ and let $F \in \mathcal{M}(F_0)$, which exists by \Cref{prop:existence}.
    For every $\pi \in S(F_0)$, the set $F \cup \{\pi\}$ is contained in $S(F_0)$ by the equality, so it is compatible by \Cref{prop:subset-closed} and contains $F_0$.
    Hence $\pi \in F$ by the maximality of $F$.
    Thus every $F \in \mathcal{M}(F_0)$ equals $S(F_0)$.
\end{proof}

By~\Cref{thm:union}, whenever the maximal compatible set from the seed is not unique, one can find seed-compatible agents that are incompatible with each other.
Such non-uniqueness is common. Even in the simple geometric setting of the motivating example, the maximal compatible set is unique only in special cases.
In fact, each maximal compatible set can be viewed as the set of agents committing to a specific convention.
For instance, two conventions prevalent worldwide are driving on the left or on the right side of the road. 
Each convention forms a maximal compatible set, yet two cars that follow different conventions are not compatible with each other.

\color{black}
\subsection{Future-proof agents are characterized by the closure of the seed}
\label{sec:closure}
As seen in the last section, being seed-compatible does not guarantee mutual compatibility, because different developers can commit to different conventions without coordination.
Therefore, in this section we introduce the concept of future-proof agents admitted by a seed.
Such agents are not only seed-compatible but also do not pre-commit to any convention.
Hence they are guaranteed to be compatible with any other future-proof agent, and they remain compatible once the regulator commits to a specific convention.

We write $\mathcal{C}$ for $\mathcal{C}_\mathcal{A}$ when the arena is clear from the context.
\begin{definition}[Future-proof agent admitted by a seed]
\label{def:future-proof}
Given a seed $F_0$, an agent $\pi\in\Pi$ is \emph{future-proof} if $\forall F' \in \mathcal{C}(F_0)$, $F' \cup \{\pi\} \in \mathcal{C}(F_0)$.
The set of all future-proof agents admitted by $F_0$ is denoted by $FP(F_0) \subseteq \Pi$. 
\end{definition}

The set of future-proof agents turns out to be characterized by the closure of the seed, which is defined as follows.
\begin{definition}[Closure of a seed]
\label{def:li-closure}
The \emph{closure} of a seed $F_0$ is the intersection of all maximal compatible sets from $F_0$, i.e.
\begin{equation*}
    Cl(F_0) \coloneqq \bigcap_{F \in \mathcal{M}(F_0)} F.
\end{equation*}
\end{definition}
\begin{theorem}[The closure characterizes future-proof agents]
\label{thm:future-proof-closure}
Given a seed $F_0$, an agent $\pi \in \Pi$ is future-proof if and only if $\pi \in Cl(F_0)$, i.e.\ $FP(F_0) = Cl(F_0)$.
\end{theorem}
\color{black}
\begin{proof}
    ($\impliedby$) Let $\pi \in Cl(F_0)$ and $F' \in \mathcal{C}(F_0)$.
    By~\Cref{thm:union} there exists $F^\star \in \mathcal{M}(F_0)$ with $F' \subseteq F^\star$.
    Since $\pi \in Cl(F_0) \subseteq F^\star$, $F' \cup \{\pi\} \subseteq F^\star$, which implies that $F' \cup \{\pi\}$ is compatible by~\Cref{prop:subset-closed}.
    Furthermore, we have $F_0 \subseteq F'$, which means $F' \cup \{\pi\} \in \mathcal{C}(F_0)$.
    Since $F'$ is arbitrary, $\pi \in FP(F_0)$.

    ($\implies$) Let $\pi \in FP(F_0)$ and $F^\star \in \mathcal{M}(F_0)$.
    Since $F^\star \in \mathcal{C}(F_0)$, \Cref{def:future-proof} gives $F^\star \cup \{\pi\} \in \mathcal{C}(F_0)$,
    and the maximality of $F^\star$ implies $\pi \in F^\star$.
    Since $F^\star$ is arbitrary, $\pi \in Cl(F_0)$.
\end{proof}
\color{black}
\Cref{def:li-closure} defines a standard closure operator on the space of compatible sets. 
Below we prove its monotonicity. The remaining properties of a closure operator are proven in \Cref{app:proofs}.
\begin{property}[Closure of a seed is monotone]
\label{prop:closure-monotone}
For any compatible sets $F_0 \subseteq F_0' \subseteq \Pi$, we have $Cl(F_0) \subseteq Cl(F_0')$.
\end{property}
\begin{proof}
    We first show $\mathcal{M}(F_0') \subseteq \mathcal{M}(F_0)$.
    Let $F \in \mathcal{M}(F_0')$. Then $F$ is compatible and $F_0 \subseteq F_0' \subseteq F$, so $F \in \mathcal{C}(F_0)$.
    Since $F$ is maximal, for every $\pi \in \Pi \setminus F$, the set $F \cup \{\pi\}$ is not compatible, hence $F \cup \{\pi\} \notin \mathcal{C}(F_0)$.
    Thus $F \in \mathcal{M}(F_0)$.
    An intersection over a smaller collection of sets is larger, so
    \begin{equation*}
        Cl(F_0) = \bigcap_{F \in \mathcal{M}(F_0)} F \subseteq \bigcap_{F \in \mathcal{M}(F_0')} F = Cl(F_0').
    \end{equation*}
\end{proof}
The monotonicity explains the observation in \Cref{sec:motivating}, where extending the seed admits more future-proof agents.
From the perspective of conventions, extending the seed means that the regulator has committed to a more specific convention.
This rules out the maximal compatible sets that do not contain the new seed agents, and the intersection of the remaining sets contains more future-proof agents.

Compared with the behavioral contract approach, our setup coordinates conventions through instances of seed agents instead of a specification of behaviors.
Nevertheless, the desired end product is in both cases a set of compatible agents.
The difference is that a behavioral contract is designed to commit to a particular convention, so that all agents satisfying it form a compatible set.
In many environments, designing a particular convention in advance can be challenging given the existing agents, and it can be overly restrictive for future agent designs.
Therefore, in our setting we leave the conventions implicit in the regulator's choice of seed,
and we keep open the possibility of extending the seed to commit to a more specific convention.

\color{black}
\subsection{Future-proof single-agent tests}
\label{sec:sat}
In this section we switch the perspective to the developer and discuss when, and to what extent, developers are able to test the future-proofness of their agents without coordinating with other developers.
Specifically, we define seed-only tests, whose results are determined by interaction with the seed alone, and we discuss their fundamental testing loss.

A seed-only test consists of a set of single-agent experiments.
\color{black}
\begin{definition}[Single-agent experiment]
\label{def:experiment}
    Let $\bullet$ be a placeholder. Given a seed $F_0$, the set of \emph{single-agent experiments} is $\mathcal{E}_{F_0} \coloneqq \mathrm{Arr}(F_0 \cup \{\bullet\})$.
    For $e \in \mathcal{E}_{F_0}$ and $\pi \in \Pi$, $e[\bullet \mapsto \pi] \in \mathrm{Arr}(F_0 \cup \{\pi\})$ denotes the arrangement obtained from $e$ by replacing every occurrence of $\bullet$ with $\pi$.
\end{definition}
\color{black}
Two agents that yield the same outcome in every single-agent experiment cannot be told apart by the seed.
We call such agents $F_0$-equivalent.
\color{black}
\begin{definition}[$F_0$-equivalence]
\label{def:equivalence}
    Given a seed $F_0$, two agents $\pi, \pi' \in \Pi$ are \emph{$F_0$-equivalent}, written $\pi \sim_{F_0} \pi'$, iff
    \begin{equation*}
        \Phi\big(e[\bullet \mapsto \pi]\big) = \Phi\big(e[\bullet \mapsto \pi']\big) \qquad \forall e \in \mathcal{E}_{F_0}.
    \end{equation*}
    The equivalence class of $\pi$ under $\sim_{F_0}$ is denoted by $[\pi]_{F_0}$.
\end{definition}
\color{black}
Under a seed-only test, all $F_0$-equivalent agents are indistinguishable, so the same test must either admit all of them or reject all of them.
\begin{definition}[Seed-only test (SOT)]
\label{def:st}
    Given a seed $F_0$, a \emph{seed-only test} is a predicate $t \colon \Pi \to \{\top, \bot\}$ that satisfies
    \begin{equation*}
        \pi \sim_{F_0} \pi' \implies t(\pi) = t(\pi') \qquad \forall \pi, \pi' \in \Pi.
    \end{equation*}
\end{definition}

Among all possible designs of SOTs, we focus on those that are sound and have a small testing loss.
We call a sound SOT a \emph{future-proof SOT}, since every agent it admits is future-proof.
\begin{definition}[Admissible set of SOT]
\label{def:admissible}
    Given a seed $F_0$ and a SOT $t$, the \emph{admissible set} $F_t$ of $t$ is defined as
    \begin{equation*}
        F_t \coloneqq \{ \pi \in \Pi \mid t(\pi) = \top \}.
    \end{equation*}
\end{definition}
\begin{definition}[Future-proof SOT]
\label{def:future-proof-sot}
    Given a seed $F_0$, a SOT $t$ is \emph{future-proof} iff
    $F_t \subseteq FP(F_0)$.
\end{definition}

The testing loss of a future-proof SOT consists of the future-proof agents that it fails to admit.
It depends on the design of the SOT. For instance, an extremely conservative SOT $t$ that rejects all agents is future-proof, but its testing loss $\mathcal{L}_t$ is the entire set $FP(F_0)$.
Here we discuss the lower bound on the testing loss over all possible SOT designs,
which is determined by the seed and the arena alone.
\begin{definition}[Testing loss of a SOT]
\label{def:testing-loss}
    Given a seed $F_0$ and a SOT $t$, the \emph{testing loss} of $t$ is $\mathcal{L}_t \coloneqq FP(F_0) \setminus F_t$.
\end{definition}
\begin{definition}[Separable future-proof agents]
\label{def:separable}
    Given a seed $F_0$, the set $Separable(F_0) \subseteq FP(F_0)$ of \emph{separable future-proof agents} consists of all agents $\pi \in FP(F_0)$ such that $[\pi]_{F_0} \subseteq FP(F_0)$.
\end{definition}
\begin{proposition}[The separable set lower-bounds the testing loss]
\label{prop:separable}
    Given a seed $F_0$, every future-proof SOT $t$ satisfies $\mathcal{L}_t \supseteq FP(F_0) \setminus Separable(F_0)$.
\end{proposition}
\begin{proof}
    If a future-proof SOT $t$ admits $\pi$, then it admits $[\pi]_{F_0}$ by~\Cref{def:st}, so $[\pi]_{F_0} \subseteq F_t \subseteq FP(F_0)$ by~\Cref{def:future-proof-sot}.
    Hence $\pi \in Separable(F_0)$.
\end{proof}
\Cref{prop:separable} shows that the separable set is the best that any future-proof SOT can achieve, and it is a fundamental property of the arena and the seed.
Specifically, the lower bound depends on the outcome space. Refining the outcomes, for example collecting the trajectories instead of a Boolean indicator of safe or not, refines the equivalence classes and can only enlarge the separable set.
Case study 1 shows an example of the arena where any seed gives $Separable(F_0) = FP(F_0)$, so zero-loss SOT can in theory be achieved for any seed.
In other cases, such as case study 2, the testing loss also depends on the seed selection. Intuitively, a seed whose experiments elicit more diverse interactions separates more agents.

We leave seed selection and practical SOT design as future works, and we discuss them in~\Cref{sec:conclusion}.

%% file: figs/closure_growth.tex
%

\def\PanelSep{3.0}
\def\PanelCx{5.82}
\def\PanelCdy{-0.38}
\def\PanelAdy{-0.2887}
\def\LblApzero{(-0.22,-0.19)}
\def\LblBpzero{(-0.70,-0.46)}\def\LblBpone{(0.70,-0.46)}
\def\LblCpzero{(-0.72,-0.46)}\def\LblCpone{(0.72,-0.46)}\def\LblCptwo{(0,0.76)}
\def\LblDeltaB{(0,-0.58)}
\def\LblDeltaC{(0,-0.58)}
\def\PairX{0.93}%
\def\LblPair{(0.62,-0.17)}%
\def\LegY{-1.55}
\def\LegX{-0.95}
\def\LegGap{2.02}
\def\SeedFont{\large}
\def\AnnFont{\normalsize}
\definecolor{cBlue}{HTML}{2C6FBB}%
\definecolor{cOrange}{HTML}{E07B1A}%
\definecolor{cClosure}{HTML}{9E7FD0}
\definecolor{cRed}{HTML}{D1495B}%
\definecolor{cGray}{HTML}{8A8A8A}%
\definecolor{cSafe}{HTML}{D9D9D9}

\begin{tikzpicture}[x=1.8cm,y=1.8cm,>=Stealth,line join=round,text=black]

  \def\Pzero{(-0.5,-0.2887)}
  \def\Pone{(0.5,-0.2887)}
  \def\Ptwo{(0,0.5774)}

  \newcommand{\seeddot}[3]{\fill[cRed] #1 circle(1.7pt);
     \node[cRed,font=\SeedFont] at #3 {#2};}

  \begin{scope}[shift={(0,\PanelAdy)}]
     \fill[cSafe,fill opacity=0.6] (0,0) circle(1);
     \draw[cGray,dashed,line width=0.6pt] (0,0) circle(1);
     \draw[cGray,line width=0.7pt] (0, 0.5) circle(0.5);              
     \draw[cOrange,line width=1.1pt,fill=cOrange,fill opacity=0.06] (-0.5,0) circle(0.5);
     \draw[cBlue,line width=1.1pt,fill=cBlue,fill opacity=0.06]     ( 0.5,0) circle(0.5);
     \fill[cClosure] (0,0) circle(2.6pt);                            
     \draw[black,dashed,line width=0.7pt] (-\PairX,0) -- (\PairX,0);
     \fill[black] (-\PairX,0) circle(1.4pt);
     \fill[black] ( \PairX,0) circle(1.4pt);
     \node[font=\AnnFont] at \LblPair {$>\delta$};
     \seeddot{(0,0)}{$p_0$}{\LblApzero}
  \end{scope}

  \begin{scope}[shift={(\PanelSep,0)}]
     \begin{scope}
        \clip \Pzero circle(1);
        \fill[cSafe,fill opacity=0.6] \Pone circle(1);
     \end{scope}
     \draw[cGray,dashed,line width=0.6pt]
        (0,0.5774) arc(120:240:1) arc(300:420:1);
     \draw[cGray,line width=0.8pt] (0,-0.2887) circle(0.5);
     \draw[cBlue,line width=1.1pt,fill=cBlue,fill opacity=0.07]
        (-0.5,-0.2887) arc(240:300:1) arc(0:60:1) arc(120:180:1) -- cycle;
     \draw[cOrange,line width=1.1pt,fill=cOrange,fill opacity=0.07]
        (-0.5,-0.2887) arc(120:60:1) arc(0:-60:1) arc(240:180:1) -- cycle;
     \draw[cClosure,line width=1.1pt,fill=cClosure,fill opacity=0.55]
        (-0.5,-0.2887) arc(120:60:1) arc(300:240:1) -- cycle;
     \draw[cRed,line width=0.9pt] \Pzero -- \Pone;
     \node[cRed,font=\AnnFont] at \LblDeltaB {$\delta$};
     \seeddot{\Pzero}{$p_0$}{\LblBpzero}
     \seeddot{\Pone}{$p_1$}{\LblBpone}
  \end{scope}

  \begin{scope}[shift={(\PanelCx,\PanelCdy)}]
     \draw[cClosure,line width=1.3pt,fill=cClosure,fill opacity=0.45]
        (-0.5,-0.2887) arc(240:300:1) arc(0:60:1) arc(120:180:1) -- cycle;
     \draw[cRed,line width=0.9pt] \Pzero -- \Pone;
     \node[cRed,font=\AnnFont] at \LblDeltaC {$\delta$};
     \seeddot{\Pzero}{$p_0$}{\LblCpzero}
     \seeddot{\Pone}{$p_1$}{\LblCpone}
     \seeddot{\Ptwo}{$p_2$}{\LblCptwo}
  \end{scope}

  \begin{scope}[shift={(\LegX,\LegY)},font=\AnnFont,text=black]
     \fill[cRed] (0,0) circle(1.7pt);
     \node[anchor=west] at (0.1,0) {seed};
     \fill[cSafe,fill opacity=0.6] ({\LegGap-0.1},-0.08) rectangle ++(0.2,0.16);
     \draw[cGray,dashed,line width=0.6pt] ({\LegGap-0.1},-0.08) rectangle ++(0.2,0.16);
     \node[anchor=west] at ({\LegGap+0.15},0) {safe with seed};
     \draw[cBlue,line width=1.1pt] ({2*\LegGap-0.1},-0.08) rectangle ++(0.13,0.13);
     \draw[cOrange,line width=1.1pt] ({2*\LegGap-0.03},-0.05) rectangle ++(0.13,0.13);
     \node[anchor=west] at ({2*\LegGap+0.15},0) {safe with each other};
     \fill[cClosure,fill opacity=0.55] ({3*\LegGap+0.2},-0.08) rectangle ++(0.2,0.16);
     \node[anchor=west] at ({3*\LegGap+0.45},0) {future-proof};
  \end{scope}

\end{tikzpicture}

%% file: sections/case-study.tex
\section{Case study}
\label{sec:experiment}
In this section, we consider two classical robotic applications, autonomous driving and warehouse delivery, 
and show that identifying compatible agent sets is non-trivial even with simple environments and agents. 
Specifically, we formulate each application within our framework and demonstrate that admitting all seed-compatible agents can lead to an incompatible set, 
and that the choice of seed strongly shapes the resulting future-proof agents. 
We further examine whether a seed admits zero testing loss and describe the zero-loss SOT when it exists.
\color{black}
\subsection{Autonomous driving: chain of vehicles on a single lane}
\label{sec:exp-braking}
\textbf{Setup.}
Consider a chain of vehicles that follow one another on a straight lane with a common velocity $v_0$ and a common bumper-to-bumper gap $g_0$, as shown in~\Cref{fig:braking-chain}.

The vehicles differ in their braking capability, which we parameterize as $\pi = (b,\tau)$, where $b \in (0, b_{\lim}]$ is the emergency braking rate and $\tau\geq 0$ is the reaction delay. Hence the agent space is the continuous set $\Pi = (0, b_{\lim}] \times \mathbb{R}_{\ge 0}$.

In such car-following scenarios, whenever the leading vehicle brakes abruptly, the following vehicle must react and brake in time to avoid a collision.
\color{black}
This safety requirement can be checked from the minimum signed gap between consecutive agents over the entire rollout,
where the gap is negative whenever the two agents overlap.
Formally, the outcome space is $\mathcal{O} = [-L, +\infty)$, where $L$ is the vehicle length,
and the set of unacceptable outcomes is $\mathcal{U} = [-L, 0)$.
\color{black}

The outcome generator $\Phi$ initiates an emergency brake of the leading vehicle and forward-simulates the resulting motion of all vehicles according to their braking capabilities.
It then records the minimum signed gap attained during the rollout.

\textbf{Seed-compatible agents do not form a compatible set.}
We begin with a single seed agent, i.e. $F_0=\{\pi^*\}$, and identify the corresponding seed-compatible agents.
Let $(\pi_1, \pi_2)$ denote a chain segment in which $\pi_2$ follows $\pi_1$.
An agent $\pi$ is compatible with $\pi^*$ if and only if the following three conditions hold.
\begin{enumerate}
  \item In the segment $(\pi^*, \pi)$, the braking distance of $\pi$ must be sufficiently short, which depends on both of its parameters.
  \item In the segment $(\pi, \pi^*)$, the braking rate of $\pi$ must be sufficiently mild for $\pi^*$ to follow safely.
  \item In the segment $(\pi, \pi)$, the two agents share the same braking rate, so $\pi$ must react before the gap is exhausted.
\end{enumerate}

The agents that satisfy all three conditions constitute the set of seed-compatible agents, as shown in~\Cref{fig:braking-single-seed-a}.

However, not all seed-compatible agents are compatible with each other.
Consider two agents $\pi_1=(b_1, \tau')$ and $\pi_2=(b_2, \tau')$ with $b_1<b_2$, where $\tau'=g_0/v_0$ is the slowest admissible reaction.
When $\pi_1$ follows $\pi_2$, a rear-end collision is inevitable, because $\pi_1$ consumes the entire gap during its reaction delay and then needs additional distance to stop.
By~\Cref{thm:union}, this counterexample implies that the maximal compatible sets are non-unique.
In fact, as \Cref{fig:braking-single-seed-b} reveals, there are infinitely many maximal compatible sets, and each of them is strictly smaller than the seed-compatible set.
Each maximal compatible set is determined by a maximal allowed braking rate and the corresponding longest braking distance,
which together can be interpreted as the convention that all of its agents commit to.
The future-proof agents, in contrast, are those that are compatible with any braking rate allowed by the seed.
In other words, they do not pre-commit to any convention.

\textbf{Future-proof agents.}
Next, we illustrate how the choice of seed shapes the future-proof agents.
Since the maximal compatible sets are determined by the maximal braking rate and the longest braking distance,
we let $B \coloneqq \sup_{F_0} b$ denote the maximal braking rate in the seed.
Likewise, we let $h(b,\tau) \coloneqq v_0\tau + v_0^2/2b$ denote the braking distance of an agent and $H \coloneqq \sup_{F_0} h$ the longest braking distance in the seed.
The closure of the seed $F_0$ then has the following expression.
\begin{equation}
  \begin{aligned}
    Cl(F_0) = \bigl\{ (b,\tau) \in \Pi \;\big|\; & b \le B, \\
    & h(b,\tau) \le \max\{H,\ g_0 + \tfrac{v_0^2}{2b_{\lim}}\} \bigr\}.
  \end{aligned}
  \label{eq:braking-closure}
\end{equation}
In words, a future-proof agent must brake no harder than any seed agent, and it must stop no later than the slowest-stopping seed agent unless it can safely follow even the strongest breaker in $\Pi$.
As the seed expands, both $B$ and $H$ are non-decreasing.
Therefore, the closure $Cl(F_0)$ is non-decreasing as well.

\Cref{fig:braking-single-seed-b,fig:braking-seed-sets-a,fig:braking-seed-sets-b} show how the maximal compatible sets and the future-proof agents evolve as the seed expands.
From~\Cref{fig:braking-single-seed-b} to \Cref{fig:braking-seed-sets-a}, the seed grows from a singleton to a set of agents.
As a result, the maximal compatible sets that do not contain the new seed are ruled out, and the set of future-proof agents grows.
The limit of this process is reached when the seed commits to a single convention.
In \Cref{fig:braking-seed-sets-b}, only one maximal compatible set remains, and every seed-compatible agent is then future-proof.

We note that a seed admitting a single maximal compatible set is desirable for testing.
However, identifying such a seed is challenging for large agent spaces and complex outcomes,
because an analytical expression of the closure, such as~\Cref{eq:braking-closure}, is usually unavailable.
Therefore, how to select a seed that represents the desired conventions remains an open question.
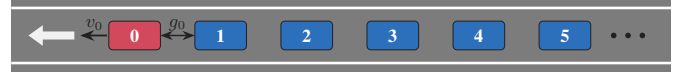
\begin{figure}[t]
  \centering
  \resizebox{\columnwidth}{!}{\input{figs/braking_chain.tex}}
  \caption{Vehicles following one another on a straight road at uniform gap $g_0$ and speed $v_0$. The outcome generator triggers an emergency stop of the leading vehicle (red) and monitors for rear-end collisions between adjacent vehicles.}
  \label{fig:braking-chain}
\end{figure}

\begin{figure}[t]
  \centering
  \begin{subfigure}[b]{0.49\columnwidth}
    \centering
    \includegraphics[width=\linewidth]{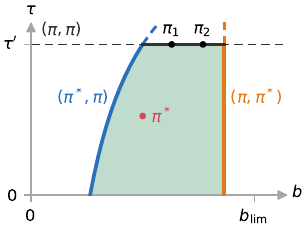}
    \caption{}
    \label{fig:braking-single-seed-a}
  \end{subfigure}
  \hfill
  \begin{subfigure}[b]{0.49\columnwidth}
    \centering
    \includegraphics[width=\linewidth]{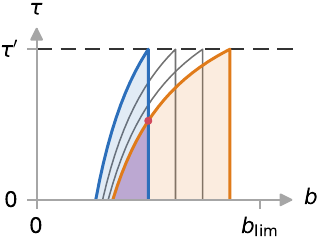}
    \caption{}
    \label{fig:braking-single-seed-b}
  \end{subfigure}
  \\[4pt]
  \begin{subfigure}[b]{0.49\columnwidth}
    \centering
    \includegraphics[width=\linewidth]{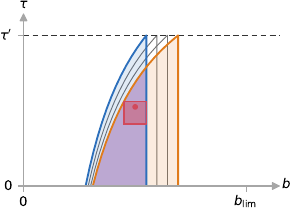}
    \caption{}
    \label{fig:braking-seed-sets-a}
  \end{subfigure}
  \hfill
  \begin{subfigure}[b]{0.49\columnwidth}
    \centering
    \includegraphics[width=\linewidth]{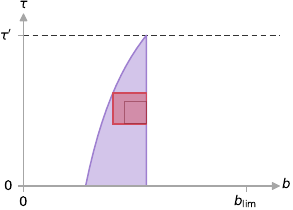}
    \caption{}
    \label{fig:braking-seed-sets-b}
  \end{subfigure}
  \caption{Seed-compatible, maximal compatible, and future-proof agents in the single-lane example.
  (a) The green region depicts all agents that are individually compatible with the seed $\pi^*$. The pair $(\pi_2, \pi_1)$ demonstrates that the set is not itself compatible.
  (b) Each outline shows a maximal compatible set, and the purple region depicts their intersection, namely the future-proof closure.
  (c) Extending the seed (red box) from $\pi^*$ in (b) admits more future-proof agents, so the purple region grows.
  (d) Further extending the seed (larger red box) collapses all maximal compatible sets to a unique one, and all seed-compatible agents become future-proof.}
  \label{fig:braking-single-seed}
\end{figure}

\textbf{SOTs with zero testing loss.}
We now show that in this example every non-empty seed admits zero testing loss.
Consider again the seed $F_0=\{\pi^*\}$.
Every seed-compatible agent is uniquely identified by the outcomes of the three experiments $(\bullet, \bullet)$, $(\pi^*, \bullet)$, and $(\bullet, \pi^*)$.
Specifically, the first outcome uniquely identifies $\tau$, and the last two are strictly monotone in $b$ once $\tau$ is known, so they uniquely identify $b$.
The same reasoning holds when the seed contains more agents.
As a result, all future-proof agents are separable for every non-empty seed, and hence SOTs with zero testing loss exist.

Let $d(\pi_1,\pi_2)$ denote the outcome of the segment $(\pi_1,\pi_2)$.
For $F_0=\{\pi^*\}$, a zero-loss SOT admits an agent $\pi$ iff
\begin{equation}
  \begin{aligned}
    & d(\pi, \pi^*) \ge d(\pi^*, \pi^*) \ \text{ and } \\
                        & d(\pi^*, \pi) \ge \min\{ d(\pi^*, \pi^*),\, (\tfrac{v_0^2}{2b^*} - \tfrac{v_0^2}{2b_{\lim}}) \}.
  \end{aligned}
  \label{eq:braking-perfect-sot}
\end{equation}
The first condition requires the admitted agent to brake no harder than the seed agent.
The second condition requires its braking distance to be no longer than that of the seed agent, or short enough to safely follow the strongest braker that can exist in the agent space.
Since the seed is compatible with itself, we have $d(\pi^*,\pi^*)\geq0$.
Therefore, \Cref{eq:braking-perfect-sot} is stricter than the seed-compatibility condition, as expected, since the latter offers no guarantee of set compatibility.

\subsection{Warehouse delivery: multi-agent path finding in grid world}
\label{sec:exp-warehouse}
\begin{figure}[t]
  \centering
  \providecommand{\flatbox}[1]{\raisebox{\depth}{#1}}%
  \resizebox{\columnwidth}{!}{%
  \begin{tabular}{@{}c@{\hspace{2pt}}c@{}}
   \resizebox{!}{5cm}{\flatbox{\begin{tabular}{@{}c@{}}%
     \input{figs/wh_map_mask.tex}\\[3pt]%
     \resizebox{7.68cm}{!}{\input{figs/wh_legend.tex}}%
   \end{tabular}}} &
   \resizebox{!}{5cm}{\flatbox{\begin{tabular}{@{}c@{}}%
     \input{figs/enc_a_resolved.tex}\\[4pt]%
     \input{figs/enc_b_collision.tex}\\[4pt]%
     \input{figs/enc_c_deadlock.tex}%
   \end{tabular}}}
  \end{tabular}}
  \caption{Warehouse delivery setup. \textbf{Left.} An example configuration. The yielding zone of the blue robot at the center is shown, and the dashed cells are the candidate starts and goals. \textbf{Right.} An encounter can lead to three outcomes, namely the conflict is resolved, the agents collide, or both agents yield.}
  \label{fig:wh-setup}
\end{figure}
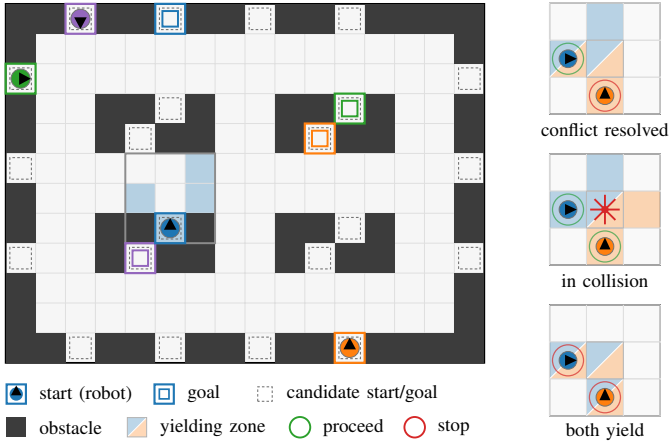
\textbf{Setup.}
We consider a set of agents operating in a grid world with obstacles, 
as illustrated in~\Cref{fig:wh-setup}. 
Each agent acts in a decentralized fashion.

All agents share the same path-planning strategy. 
At every tick, an agent moves to the adjacent cell with the lowest cost-to-go to its goal cell, 
unless it must yield to other agents. 
The agents differ in their \emph{yielding zone}, which is a subset of the eight cells surrounding the agent in the $3\times3$ block ahead of it.
An agent stops only if an unfinished agent occupies a cell in its yielding zone.
Since the yielding zone has $2^8$ possible options, the agent space has cardinality $|\Pi| = 256$.

\color{black}
As shown in~\Cref{fig:wh-setup}, an encounter between agents with different yielding zones may result in conflict resolution, a collision, or a deadlock.
We say that a combination of agents is acceptable if, in all arrangements, the agents are collision-free and reach their goals within a reasonable amount of time. 
The outcome is a pair of Booleans that indicate the occurrence of a collision and of a timeout, so $\mathcal{O} = 2^2$ and $\mathcal{U} = \mathcal{O}\setminus\{(1,1)\}$.
We set the time limit of an instance to 10 times the longest cost-to-go among its tasks.
\color{black}

The outcome generator $\Phi$ runs a synchronized simulation on a collection of maps and goal assignments, 
and it monitors the occurrence of collisions and timeouts. 
The configurations used in this case study are detailed in~\Cref{app:wh-implementation}.

Randomness is introduced to break symmetry and mitigate deadlocks.
Specifically, an agent yields for only 1, 2, or 3 ticks with equal probability, and then it moves to an adjacent free cell.
In addition, when multiple cells have the same cost-to-go, one of them is picked at random.

\textbf{Self-compatible agents are rarely compatible with one another.}
We begin by identifying all self-compatible agents. In our framework, these are the seed-compatible agents of the empty seed $F_0 = \emptyset$.
As shown in~\Cref{fig:self-compat}, 38 of the 256 agents are self-compatible.
Every one of them yields to at least one of the two front-diagonal cells.
We group them by that side into 18 \emph{yield-to-left} agents, 16 \emph{yield-to-right} agents, and 4 agents that yield to both sides.

\begin{figure}[t]
  \centering
  \includegraphics[width=\columnwidth]{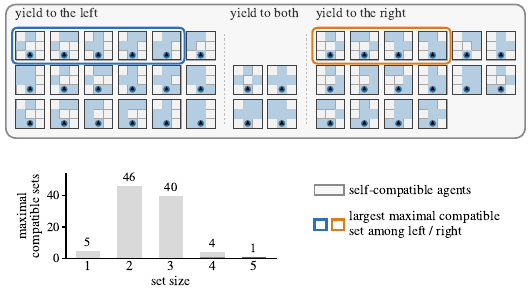}
  \caption{The 38 self-compatible agents, each drawn as a robot with its yielding zone, arranged by the side they yield to. The outlines mark the largest compatible set among the left-yielding agents (five agents) and among the right-yielding agents (four agents). No maximal compatible set contains more than five agents.}
  \label{fig:self-compat}
\end{figure}

Self-compatibility, however, does not imply mutual compatibility.
The 38 agents induce 96 maximal compatible sets, whose sizes are shown in the histogram of~\Cref{fig:self-compat}.
Most of the maximal compatible sets consist of a single pair of agents, and the largest set contains only five agents.
As a result, if all self-compatible agents are permitted in the environment, they are very likely to be incompatible with each other.

Although the compatibility results may change due to randomness, 
a collision between a yield-to-left and a yield-to-right agent is unavoidable. 
Therefore, any maximal compatible set contains agents of at most one of these two kinds. 
In other words, each maximal compatible set commits to a single yielding convention.

\begin{figure}[t]
  \centering
  \includegraphics[width=\columnwidth]{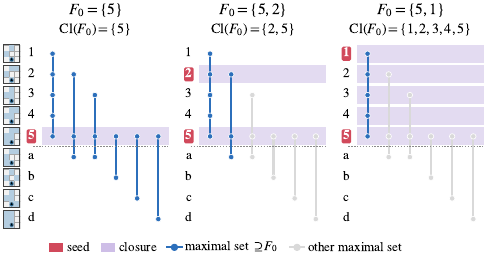}
  \caption{Maximal compatible sets containing agent 5, one per column. Agents 1--5 form the only compatible set of five, and a--d are other agents. With $F_0=\{5\}$, the closure is the seed itself. Adding agent 2 leaves two maximal sets and $\mathrm{Cl}(F_0)=\{2,5\}$, whereas adding agent 1 leaves only the five-agent set, so $\mathrm{Cl}(F_0)=\{1,\dots,5\}$.}
  \label{fig:seed-selection}
\end{figure}

\textbf{The seed determines which maximal compatible set is admitted.}
We now examine how the choice of seed shapes the future-proof agents,
using the five-agent compatible set labeled 1 to 5 in~\Cref{fig:seed-selection}.
The figure also shows all six maximal compatible sets to which agent 5 belongs.
With the single-agent seed $F_0=\{5\}$, the closure is the intersection of all plotted maximal compatible sets, which have only the seed in common.
How much the closure grows depends on which agent is added to the seed.
Adding agent 2 leaves two maximal compatible sets and enlarges the closure only to $\{2,5\}$.
Adding agent 1 instead leaves the five-agent set as the unique maximal compatible set, so the closure becomes $\{1,\dots,5\}$ and every seed-compatible agent is future-proof.
In other words, the second seed agent commits the arena to one convention, and committing to the more selective agent 1 admits more agents.

\textbf{Zero testing loss requires a unique maximal compatible set.}
In this case study, a seed admits a zero-loss SOT only if it lies in a unique maximal compatible set.
The restriction stems from the Boolean outcome.
Every seed-compatible agent yields the same outcome, namely neither a collision nor a timeout, in all experiments with the seed, so the seed-compatible agents form a single $F_0$-equivalence class.
Hence, if the seed lies in several maximal compatible sets, not all seed-compatible agents are future-proof, and none of the future-proof agents is separable from the non-future-proof ones.
Consequently, no future-proof SOT admits any agent.

There are two ways to enable a meaningful SOT. 
One is to enrich the outcome space with more information, such as the trajectories of all agents. 
The other is to choose the seed such that it admits a unique maximal compatible set.

%% file: figs/braking_chain.tex
\begin{tikzpicture}[line join=round,x=1cm,y=1cm,>=Stealth]
  \definecolor{cRed}{HTML}{D1495B}
  \definecolor{cBlue}{HTML}{2C6FBB}
  \definecolor{cDark}{HTML}{222222}
  \definecolor{cRoad}{HTML}{7C7C7C}

  \def\HL{0.45}   
  \def\HH{0.25}   

  \fill[cRoad] (-2.15,-0.64) rectangle (9.35,0.64);
  \draw[white,line width=1.3pt] (-2.15, 0.48) -- (9.35, 0.48);
  \draw[white,line width=1.3pt] (-2.15,-0.48) -- (9.35,-0.48);
  \fill[white,opacity=0.9]
      (-1.85,0) -- (-1.60,0.18) -- (-1.60,0.07) -- (-1.05,0.07)
      -- (-1.05,-0.07) -- (-1.60,-0.07) -- (-1.60,-0.18) -- cycle;

  \foreach \i/\x in {0/0,1/1.5,2/3.0,3/4.5,4/6.0,5/7.5}{
    \ifnum\i=0
      \fill[cRed,rounded corners=2pt] (\x-\HL,-\HH) rectangle (\x+\HL,\HH);
      \draw[cDark,line width=0.5pt,rounded corners=2pt] (\x-\HL,-\HH) rectangle (\x+\HL,\HH);
    \else
      \fill[cBlue,rounded corners=2pt] (\x-\HL,-\HH) rectangle (\x+\HL,\HH);
      \draw[cDark,line width=0.5pt,rounded corners=2pt] (\x-\HL,-\HH) rectangle (\x+\HL,\HH);
    \fi
    \node[white,font=\small\bfseries] at (\x,0) {\i};
  }

  \foreach \x in {8.35,8.6,8.85}{ \fill[cDark] (\x,0) circle(1.4pt); }

  \draw[->,cDark,line width=0.9pt] (-0.48,0) -- (-0.95,0);
  \node[cDark,font=\small] at (-0.70,0.19) {$v_0$};

  \draw[<->,cDark,line width=0.7pt] (0.45,0) -- (1.05,0);
  \node[cDark,font=\small] at (0.75,0.19) {$g_0$};
\end{tikzpicture}

%% file: figs/wh_map_mask.tex
\definecolor{whObs}{RGB}{60,60,60}%
\definecolor{whA}{RGB}{31,119,180}%
\definecolor{whB}{RGB}{44,160,44}%
\definecolor{whC}{RGB}{255,127,14}%
\definecolor{whD}{RGB}{148,103,189}%
\definecolor{whHi}{RGB}{214,40,40}%
\providecommand{\whhd}[3]{\begin{scope}[shift={(#1,#2)},rotate=#3]%
  \fill[black] (0,0.34)--(-0.21,-0.05)--(0.21,-0.05)--cycle;\end{scope}}%
\providecommand{\whrobot}[4]{\fill[#4] (#1+0.5,#2+0.5) circle (0.32);%
  \draw[black!55,line width=0.2pt] (#1+0.5,#2+0.5) circle (0.32);\whhd{#1+0.5}{#2+0.5}{#3}}%
\providecommand{\whstart}[3]{\draw[#3,line width=1pt] (#1,#2) rectangle ++(1,1);}%
\providecommand{\whgoal}[3]{\draw[#3,line width=1pt] (#1,#2) rectangle ++(1,1);%
  \draw[#3,line width=0.9pt] (#1+0.24,#2+0.24) rectangle ++(0.52,0.52);}%
\begin{tikzpicture}[scale=0.48]
  \fill[gray!7] (0,0) rectangle (16,12);
  \fill[whA,fill opacity=0.30] (5,4) rectangle ++(1,1);
  \fill[whA,fill opacity=0.30] (4,5) rectangle ++(1,1);
  \fill[whA,fill opacity=0.30] (6,5) rectangle ++(1,1);
  \fill[whA,fill opacity=0.30] (6,6) rectangle ++(1,1);
  \draw[gray!25,step=1] (0,0) grid (16,12);
  \draw[black,line width=0.8pt] (0,0) rectangle (16,12);
  \foreach \c/\r in {0/11,1/11,3/11,4/11,6/11,7/11,9/11,10/11,12/11,13/11,14/11,15/11,0/10,15/10,0/8,3/8,
     4/8,6/8,9/8,10/8,12/8,15/8,0/7,3/7,5/7,6/7,9/7,11/7,12/7,15/7,0/5,15/5,
     0/4,3/4,4/4,6/4,9/4,10/4,12/4,15/4,3/3,5/3,6/3,9/3,11/3,12/3,0/2,15/2,
     0/1,15/1,0/0,1/0,3/0,4/0,6/0,7/0,9/0,10/0,12/0,13/0,14/0,15/0}
     \fill[whObs] (\c,\r) rectangle ++(1,1);
  \draw[black!45,line width=0.9pt] (4,4) rectangle (7,7);
  \foreach \c/\r in {2/11,5/11,8/11,11/11,0/9,15/9,5/8,11/8,4/7,10/7,0/6,15/6,5/4,11/4,0/3,4/3,10/3,15/3,2/0,5/0,8/0,11/0}
     \draw[black!55,line width=0.5pt,dash pattern=on 1.2pt off 1pt] (\c+0.14,\r+0.14) rectangle ++(0.72,0.72);
  \whgoal{5}{11}{whA}\whgoal{11}{8}{whB}\whgoal{10}{7}{whC}\whgoal{4}{3}{whD}
  \whstart{0}{9}{whB}\whrobot{0}{9}{-90}{whB}
  \whstart{11}{0}{whC}\whrobot{11}{0}{0}{whC}
  \whstart{2}{11}{whD}\whrobot{2}{11}{180}{whD}
  \whstart{5}{4}{whA}\whrobot{5}{4}{0}{whA}
\end{tikzpicture}

%% file: figs/wh_legend.tex
\definecolor{whObs}{RGB}{60,60,60}%
\definecolor{whA}{RGB}{31,119,180}%
\definecolor{whB}{RGB}{44,160,44}%
\definecolor{whC}{RGB}{255,127,14}%
\definecolor{whHi}{RGB}{214,40,40}%
\begin{tikzpicture}
  \draw[whA,line width=1pt] (0,0.55) rectangle ++(0.32,0.32);
  \fill[whA] (0.16,0.71) circle (0.11);\begin{scope}[shift={(0.16,0.71)},scale=0.42]\fill[black] (0,0.34)--(-0.21,-0.05)--(0.21,-0.05)--cycle;\end{scope}
  \node[anchor=west,font=\small] at (0.42,0.71) {start (robot)};
  \draw[whA,line width=1pt] (2.45,0.55) rectangle ++(0.32,0.32);
  \draw[whA,line width=0.8pt] (2.53,0.63) rectangle ++(0.16,0.16);
  \node[anchor=west,font=\small] at (2.87,0.71) {goal};
  \draw[black!55,line width=0.5pt,dash pattern=on 1.2pt off 1pt] (4.16,0.59) rectangle ++(0.24,0.24);
  \node[anchor=west,font=\small] at (4.52,0.71) {candidate start/goal};
  \fill[whObs] (0,0) rectangle ++(0.32,0.32);
  \node[anchor=west,font=\small] at (0.42,0.16) {obstacle};
  \fill[whA,fill opacity=0.30] (2.0,0)--(2.32,0.32)--(2.0,0.32)--cycle;
  \fill[whC,fill opacity=0.30] (2.0,0)--(2.32,0)--(2.32,0.32)--cycle;
  \draw[white,line width=0.4pt] (2.0,0)--(2.32,0.32);
  \draw[gray!60,line width=0.3pt] (2.0,0) rectangle ++(0.32,0.32);
  \node[anchor=west,font=\small] at (2.42,0.16) {yielding zone};
  \draw[whB,line width=0.9pt] (4.86,0.16) circle (0.17);
  \node[anchor=west,font=\small] at (5.12,0.16) {proceed};
  \draw[whHi,line width=0.9pt] (6.76,0.16) circle (0.17);
  \node[anchor=west,font=\small] at (7.02,0.16) {stop};
\end{tikzpicture}

%% file: figs/enc_a_resolved.tex
\definecolor{whA}{RGB}{31,119,180}   
\definecolor{whC}{RGB}{255,127,14}   
\definecolor{gogc}{RGB}{46,160,67}   
\definecolor{stopc}{RGB}{214,40,40}  
\providecommand{\rbt}[4]{\fill[#4] (#1,#2) circle (0.24);\draw[black!55,line width=0.2pt](#1,#2) circle(0.24);%
  \begin{scope}[shift={(#1,#2)},rotate=#3]\fill[black](0,0.24)--(-0.15,-0.08)--(0.15,-0.08)--cycle;\end{scope}}
\providecommand{\gridbg}{\fill[gray!5](0,0)rectangle(3,3);\draw[gray!25,step=1](0,0)grid(3,3);%
  \draw[black,line width=0.7pt](0,0)rectangle(3,3);\draw[gray!55,line width=0.6pt](1,1)rectangle(2,2);}
\providecommand{\mcell}[3]{\fill[#1,fill opacity=0.30] (#2,#3) rectangle ++(1,1);}%
\providecommand{\msplit}[2]{\fill[whA,fill opacity=0.30] (#1,#2)--(#1+1,#2+1)--(#1,#2+1)--cycle;%
  \fill[whC,fill opacity=0.30] (#1,#2)--(#1+1,#2)--(#1+1,#2+1)--cycle;%
  \draw[white,line width=0.4pt] (#1,#2)--(#1+1,#2+1);}
\providecommand{\goring}[2]{\draw[gogc,line width=0.8pt,draw opacity=0.7] (#1,#2) circle (0.42);}%
\providecommand{\stopring}[2]{\draw[stopc,line width=0.8pt,draw opacity=0.7] (#1,#2) circle (0.42);}%
\providecommand{\boom}{\foreach \a in {0,45,90,135,180,225,270,315}\draw[stopc,line width=1.2pt](1.5,1.5)--++(\a:0.40);\fill[stopc](1.5,1.5)circle(0.10);}

\begin{tikzpicture}[x=0.82cm,y=0.82cm]
  \gridbg
  \mcell{whA}{1}{2}
  \mcell{whC}{1}{0}
  \msplit{0}{1}
  \msplit{1}{1}
  \draw[gray!25,step=1](0,0)grid(3,3);\draw[gray!55,line width=0.6pt](1,1)rectangle(2,2);
  \rbt{0.5}{1.5}{-90}{whA}
  \rbt{1.5}{0.5}{0}{whC}
  \goring{0.5}{1.5}
  \stopring{1.5}{0.5}
  \node[font=\large] at (1.5,-0.42) {conflict resolved};
\end{tikzpicture}

%% file: figs/enc_b_collision.tex
\definecolor{whA}{RGB}{31,119,180}   
\definecolor{whC}{RGB}{255,127,14}   
\definecolor{gogc}{RGB}{46,160,67}   
\definecolor{stopc}{RGB}{214,40,40}  
\providecommand{\rbt}[4]{\fill[#4] (#1,#2) circle (0.24);\draw[black!55,line width=0.2pt](#1,#2) circle(0.24);%
  \begin{scope}[shift={(#1,#2)},rotate=#3]\fill[black](0,0.24)--(-0.15,-0.08)--(0.15,-0.08)--cycle;\end{scope}}
\providecommand{\gridbg}{\fill[gray!5](0,0)rectangle(3,3);\draw[gray!25,step=1](0,0)grid(3,3);%
  \draw[black,line width=0.7pt](0,0)rectangle(3,3);\draw[gray!55,line width=0.6pt](1,1)rectangle(2,2);}
\providecommand{\mcell}[3]{\fill[#1,fill opacity=0.30] (#2,#3) rectangle ++(1,1);}%
\providecommand{\msplit}[2]{\fill[whA,fill opacity=0.30] (#1,#2)--(#1+1,#2+1)--(#1,#2+1)--cycle;%
  \fill[whC,fill opacity=0.30] (#1,#2)--(#1+1,#2)--(#1+1,#2+1)--cycle;%
  \draw[white,line width=0.4pt] (#1,#2)--(#1+1,#2+1);}
\providecommand{\goring}[2]{\draw[gogc,line width=0.8pt,draw opacity=0.7] (#1,#2) circle (0.42);}%
\providecommand{\stopring}[2]{\draw[stopc,line width=0.8pt,draw opacity=0.7] (#1,#2) circle (0.42);}%
\providecommand{\boom}{\foreach \a in {0,45,90,135,180,225,270,315}\draw[stopc,line width=1.2pt](1.5,1.5)--++(\a:0.40);\fill[stopc](1.5,1.5)circle(0.10);}

\begin{tikzpicture}[x=0.82cm,y=0.82cm]
  \gridbg
  \mcell{whA}{0}{1}
  \mcell{whA}{1}{2}
  \mcell{whC}{1}{0}
  \mcell{whC}{2}{1}
  \msplit{1}{1}
  \draw[gray!25,step=1](0,0)grid(3,3);\draw[gray!55,line width=0.6pt](1,1)rectangle(2,2);
  \rbt{0.5}{1.5}{-90}{whA}
  \rbt{1.5}{0.5}{0}{whC}
  \boom
  \goring{0.5}{1.5}
  \goring{1.5}{0.5}
  \node[font=\large] at (1.5,-0.42) {in collision};
\end{tikzpicture}

%% file: figs/enc_c_deadlock.tex
\definecolor{whA}{RGB}{31,119,180}   
\definecolor{whC}{RGB}{255,127,14}   
\definecolor{gogc}{RGB}{46,160,67}   
\definecolor{stopc}{RGB}{214,40,40}  
\providecommand{\rbt}[4]{\fill[#4] (#1,#2) circle (0.24);\draw[black!55,line width=0.2pt](#1,#2) circle(0.24);%
  \begin{scope}[shift={(#1,#2)},rotate=#3]\fill[black](0,0.24)--(-0.15,-0.08)--(0.15,-0.08)--cycle;\end{scope}}
\providecommand{\gridbg}{\fill[gray!5](0,0)rectangle(3,3);\draw[gray!25,step=1](0,0)grid(3,3);%
  \draw[black,line width=0.7pt](0,0)rectangle(3,3);\draw[gray!55,line width=0.6pt](1,1)rectangle(2,2);}
\providecommand{\mcell}[3]{\fill[#1,fill opacity=0.30] (#2,#3) rectangle ++(1,1);}%
\providecommand{\msplit}[2]{\fill[whA,fill opacity=0.30] (#1,#2)--(#1+1,#2+1)--(#1,#2+1)--cycle;%
  \fill[whC,fill opacity=0.30] (#1,#2)--(#1+1,#2)--(#1+1,#2+1)--cycle;%
  \draw[white,line width=0.4pt] (#1,#2)--(#1+1,#2+1);}
\providecommand{\goring}[2]{\draw[gogc,line width=0.8pt,draw opacity=0.7] (#1,#2) circle (0.42);}%
\providecommand{\stopring}[2]{\draw[stopc,line width=0.8pt,draw opacity=0.7] (#1,#2) circle (0.42);}%
\providecommand{\boom}{\foreach \a in {0,45,90,135,180,225,270,315}\draw[stopc,line width=1.2pt](1.5,1.5)--++(\a:0.40);\fill[stopc](1.5,1.5)circle(0.10);}

\begin{tikzpicture}[x=0.82cm,y=0.82cm]
  \gridbg
  \msplit{0}{1}
  \msplit{1}{0}
  \msplit{1}{1}
  \draw[gray!25,step=1](0,0)grid(3,3);\draw[gray!55,line width=0.6pt](1,1)rectangle(2,2);
  \rbt{0.5}{1.5}{-90}{whA}
  \rbt{1.5}{0.5}{0}{whC}
  \stopring{0.5}{1.5}
  \stopring{1.5}{0.5}
  \node[font=\large] at (1.5,-0.42) {both yield};
\end{tikzpicture}

%% file: sections/conclusion.tex
\section{Conclusion}
\label{sec:conclusion}
This paper establishes a theory of compatibility for multi-agent systems.
The framework is general in three respects.
First, it applies to any multi-agent system and any specification, since an arena requires only an agent space, an outcome generator, and a set of unacceptable outcomes.
Second, it is agnostic to agent design, because agents enter the theory solely through the outcomes they produce.
Third, it accounts for both the agents that exist today and the agents admitted in the future, which is the setting in which deployed systems evolve.
Within this framework, we formalized compatibility, characterized the future-proof agents permitted by a seed as its closure, and introduced seed-only tests together with a lower bound on their testing loss.

The results open two research directions.
The first is seed selection.
The seed determines not only the extent of the closure but also the testing loss, as both case studies demonstrate.
A principled method for choosing a seed that admits many future-proof agents while keeping the testing loss small is therefore of immediate practical value.
The second is the design of future-proof SOTs.
The separable set bounds what any future-proof SOT can admit, yet computing it requires the closure and is in general intractable.
Practical SOTs must therefore approximate the separable set from inside, and the central question is how to trade permissiveness against computational cost.

Throughout this work, the agent space is fixed, whereas in practice it grows as new designs appear.
For a given seed, such growth can only reduce the existing future-proof agents,
because new agents may form additional maximal compatible sets with which some existing future-proof agents are incompatible.
The same holds for the separable future-proof agents.
Indeed, a new agent may be $F_0$-equivalent to an existing future-proof agent without being future-proof itself.
Enlarging the seed counters both effects, as it enlarges the closure and refines the equivalence classes.
We therefore regard the stability of the closure and the separable set under an expanding agent space as part of the seed selection problem, which we leave for future work.

%% file: sections/appendix.tex
\appendices
\crefalias{section}{appendix}

\section{Proofs}
\label{app:proofs}

\begin{proof}[Proof of \Cref{prop:subset-closed}]
    $\mathrm{Arr}(F) \subseteq \mathrm{Arr}(G)$, so $\Phi(a) \cap \mathcal{U} = \emptyset$ holds for every $a \in \mathrm{Arr}(F)$.
\end{proof}

\begin{proof}[Proof of \Cref{prop:existence}]
    The idea is to grow $F_0$ by adding agents as long as compatibility is preserved.
    By Zorn's lemma, this process terminates at a maximal set as long as every chain of growing compatible sets has a compatible upper bound.

    Order $\mathcal{C}(F_0)$ by inclusion, which is non-empty since $F_0 \in \mathcal{C}(F_0)$.
    Let $\{F_i\}_{i \in I} \subseteq \mathcal{C}(F_0)$ be a chain.
    If $I = \emptyset$, then $F_0$ is an upper bound of the chain.
    Otherwise, let $F^\infty \coloneqq \bigcup_{i \in I} F_i$ and let $a = (a_1, \dots, a_n) \in \mathrm{Arr}(F^\infty)$.
    For each entry $a_k$, there exists $i_k \in I$ such that $a_k \in F_{i_k}$.
    Since the chain is totally ordered, the largest of the $n$ sets $F_{i_1}, \dots, F_{i_n}$ contains all entries, and we denote it by $F_{i^\ast}$.
    Since $F_{i^\ast}$ is compatible and $a \in \mathrm{Arr}(F_{i^\ast})$, we have $\Phi(a) \cap \mathcal{U} = \emptyset$.
    As $a$ is arbitrary, $F^\infty$ is compatible, and $F_0 \subseteq F^\infty$ gives $F^\infty \in \mathcal{C}(F_0)$.
    Hence every chain in $\mathcal{C}(F_0)$ has an upper bound in $\mathcal{C}(F_0)$, and by Zorn's lemma $\mathcal{C}(F_0)$ has a maximal element $F^\star$ with respect to inclusion.
    That is, $F^\star \cup \{\pi\} \notin \mathcal{C}(F_0)$ for every $\pi \in \Pi \setminus F^\star$, so $F^\star$ is a maximal compatible set from $F_0$ by \Cref{def:maximal-pop}.
\end{proof}

\begin{property}[Future-proof closure of a seed is a closure operator]
\label{prop:closure-operator}
The Future-proof closure operator $Cl(\cdot)$ is a closure operator on the space of compatible sets $\mathcal{C}$, i.e. it satisfies
\begin{enumerate}
    \item[(C1)] $\forall F_0 \in \mathcal{C},\ F_0 \subseteq Cl(F_0)$
    \item[(C2)] $\forall F_0, F_1 \in \mathcal{C},\ F_0 \subseteq F_1 \Rightarrow Cl(F_0) \subseteq Cl(F_1)$
    \item[(C3)] $\forall F_0 \in \mathcal{C},\ Cl(Cl(F_0)) = Cl(F_0)$
\end{enumerate}
\end{property}
\begin{proof}[Proof of \Cref{prop:closure-operator} (closure operator)]
    (C1) holds because $F_0 \subseteq F$ for all $F \in \mathcal{M}(F_0)$, so $F_0 \subseteq \bigcap_{F \in \mathcal{M}(F_0)} F = Cl(F_0)$.
    
    (C2) holds by \Cref{prop:closure-monotone}.

    (C3) holds Let $F_0 \subseteq \Pi$ be a compatible set. We denote $Cl(F_0)=G_0$.
$\forall F\in\mathcal{M}(F_0)$, we have that $F$ is a maximal compatible set and $G_0\subseteq F$. Thus $F\in \mathcal{M}(G_0)$.
Similarly, $\forall G\in\mathcal{M}(G_0)$, we have that $G$ is a maximal compatible set and $F_0\subseteq G_0\subseteq G$. Thus $G\in \mathcal{M}(F_0)$.
Therefore, $\mathcal{M}(F_0) = \mathcal{M}(G_0)$, and we have $Cl(F_0) = Cl(G_0)=Cl(Cl(F_0))$.
\end{proof}





\section{Implementation details of~\Cref{sec:exp-warehouse}}
\label{app:wh-implementation}
The map collection three corridor scenarios, which size $12\times16$, $16\times22$ and $20\times34$, and 4, 5 and 6 agents respectively. 
360 instances are created by randomly select the start and goal assignments from the pockets on one of the three maps. 
The trivial instances where blind agents can solve are filtered out and not counted in the 360 instances.